%% file: jxhc.tex
\documentclass[10pt,conference]{IEEEtran}
\newcommand{\CLASSINPUTinnersidemargin}{1in} % inner side margin
\newif\ifpdf
\ifx\pdfoutput\undefined \else
  \ifx\pdfoutput\relax
  \else
    \ifcase\pdfoutput
    \else
      \pdftrue
    \fi
  \fi
\fi

\usepackage{cite}

\ifCLASSINFOpdf
   \usepackage[pdftex]{graphicx}
   \graphicspath{{.}}
 \else
   \usepackage[dvips]{graphicx}
   \graphicspath{{.}}
 \fi

\usepackage{nicefrac}
\usepackage{array}
\usepackage{mdwmath}
\usepackage{mdwtab}
 
\IEEEoverridecommandlockouts
\def\Note(#1){{\color{blue}(#1)}}
\usepackage{bm,amsmath,amssymb,newproof,nicefrac}
\usepackage{amsmath, amsfonts, epsfig, color, url}
\usepackage{nicefrac,times, graphicx,amsmath,amsthm,amssymb}
\usepackage[tight,footnotesize]{subfigure}
\usepackage{bm,amsmath,amssymb,newproof,nicefrac}
\usepackage{algorithmic}
\usepackage{algorithm}
\newcommand{\n}{{\cal C}}
\newcommand{\m}{{ N}}

\newcommand{\A}{{\Phi}}
\newcommand{\srm}{Delsarte-Goethals Codes~}
\newcommand{\Ex}{\mathbb{E}}
\newcommand{\nor}{\frac{1}{\sqrt{\mm}}}
\newcommand{\pii}[1]{ {P_{#1},b_{#1} } }

\newcommand{\bs}[1]{ {\boldsymbol{#1} } }
\newcommand{\g}{\Gamma_a^{\ell}}

\newtheorem{theorem}{Theorem}

\newtheorem{lemma}[theorem]{Lemma}
\theoremstyle{remark}

\newtheorem{definition}{Definition}

\newcommand{\as}{\boldsymbol{\alpha}}
\newcommand{\has}{{\hat{\as}}}

\newcommand{\mm}{N}
\newcommand{\f}{f}
\newcommand{\at}[1]{#1(x+a)\overline{#1(x)}}

\usepackage{setspace}
\usepackage{color}
\begin{document}
\title{{A Sublinear Algorithm for Sparse Reconstruction with $\boldsymbol{\ell_2 / \ell_2}$ Recovery Guarantees}}

\author{
\IEEEauthorblockN{
Robert Calderbank\thanks{The work of R. Calderbank and S. Jafarpour is supported in part by NSF under grant DMS 0701226, by ONR under grant N00173-06-1-G006, and by
AFOSR under grant FA9550-05-1-0443.
}
}
\IEEEauthorblockA{
Mathematics \& Electrical Engineering\\
Princeton University\\
 NJ 08544, USA
}
\and
\IEEEauthorblockN{Stephen Howard}
\IEEEauthorblockA{DSTO\\
PO Box 1500\\ 
Edinburgh 5111, Australia
}
\and
\IEEEauthorblockN{Sina Jafarpour}
\IEEEauthorblockA{Computer Science\\
Princeton University\\ 
 NJ 08544, USA}
 }

\maketitle
\begin{abstract}
Compressed Sensing aims to capture attributes of a sparse signal using very few measurements. Cand\`{e}s and Tao showed that sparse reconstruction is possible if the sensing matrix acts as a near isometry on all $\boldsymbol{k}$-sparse signals. This property holds with overwhelming probability if the entries of the matrix are generated by an iid Gaussian or Bernoulli process. There has been significant recent interest in an alternative signal processing framework; exploiting deterministic sensing matrices that with overwhelming probability act as a near isometry on $\boldsymbol{k}$-sparse vectors with uniformly random support, a geometric condition that is called the Statistical Restricted Isometry Property or StRIP. This paper considers a family of deterministic sensing matrices satisfying the StRIP that are based on \srm codes (binary chirps) and a $\boldsymbol{k}$-sparse reconstruction algorithm with sublinear complexity. 
{
In the presence of stochastic noise in the data domain, this paper derives bounds on the $\boldsymbol{\ell_2}$ accuracy of approximation in terms of the $\boldsymbol{\ell_2}$ norm of the measurement noise and the accuracy of the best $\boldsymbol{k}$-sparse approximation, also measured in the $\boldsymbol{\ell_2}$ norm. This type of $\boldsymbol{\ell_2 /\ell_2}$ bound is tighter than the standard $\boldsymbol{\ell_2 /\ell_1}$ or $\boldsymbol{\ell_1/ \ell_1}$ bounds.  
}
\end{abstract}
\section{Introduction}
\label{sec:intro}
\input{intro}
\section{\srm}
\label{sec:matrices}
\input{srm}
\section{The Chirp Reconstruction Algorithm}
\label{sec:quad}
\input{quad-alg}

\section{Analysis of the Algorithm}

\label{sec:anaquad}
\input{matrices}

\label{sec:ell}
\input{l2}
\bibliographystyle{IEEEbib} 
\bibliography{jxhc} 
\end{document}

%% file: intro.tex
The central goal of \textit{compressed sensing} is to capture attributes of a signal using very few measurements. In most work to date, this broader objective is exemplified by the important special case in which a $k$-\textit{sparse} vector $\as$ in $\mathbb{R}^{\n}$ with $\n$ large is to be reconstructed from a small number $\m$ of linear measurements with $k < \m \ll \n$. In this problem, the measurement data is a vector $f =\Phi \as$, where $\Phi$ is an $\m \times \n$ matrix called the \textit{sensing matrix}. 

The work of Donoho \cite{Donoho} and of Cand\`{e}s, Romberg and Tao \cite{CRT2} provides fundamental insight into the geometry of sensing matrices. The \textit{Restricted Isometry Property} (RIP) formulated by Cand\`{e}s and Tao \cite{CT} is that the sensing matrix acts as a near isometry on all $k$-sparse vectors, and this condition is sufficient for sparse reconstruction. There are two broad families of reconstruction algorithms, those based on convex optimization and those based on greedy iteration. The basis pursuit algorithms try to find the sparse approximation by relaxing the non-convex $\ell_0$ loss to a convex optimization task such as $\ell_1$ minimization, and LASSO \cite{CRT2}. The Matching Pursuit algorithms \cite{GT,DM,NT} on the other hand try to solve the recovery problem iteratively. At each iteration, one or a list of coordinates is selected greedily to provide the best approximation to the vector in the measurement domain. The vector in the measurement domain is then updated accordingly at the end of each iteration. Adjacency matrices of expander graphs have been shown to provide similar performance \cite{BGIST,sina,IR}.

One disadvantage of these Basis Pursuit and Matching Pursuit algorithms is that computational complexity is super-linear in the dimension of the data domain, which is typically very large if $k\ll \n$. In this paper, focusing on average case performance, we propose and analyze a Chirp Reconstruction Algorithm that reconstructs a $k$-sparse vector iteratively by forming the power spectrum of the measured superposition. By contrast the complexity of Chirp Reconstruction depends only on the sparsity level k and the number of measurements $\m$. {A second disadvantage is that even though reconstructing a $k$-sparse signal in the presence of noise in the data-domain is a fundamentally important problem, bounds on the accuracy of approximation of BP and MP algorithms are not very tight. Let $\as_k$ be $\as$ restricted to its $k$ most significant entries, $\mu$ be the noise vector, and $\has^*$ be the output of the recovery algorithm. An algorithm is said to provide $\ell_p/\ell_q$ recovery guarantees if $$\|\as-\has^*\|_p \leq C_1(k)\|\as-\as_k\|_q+C_2\|\mu\|_p.$$ The sparse reconstruction algorithms that use random dense matrices provide $\ell_2/\ell_1$ guarantees, and the expander-based reconstruction algorithms provide $\ell_1/\ell_1$ guarantees. The reason again goes to the worst-case vs stochastic modeling of the noise in the data domain. A result by Cohen \textit{et. al} \cite{best} shows that no reconstruction algorithm can provide $\ell_2/\ell_2$ reconstruction guarantees unless $\mm=\Omega(\n)$. Nevertheless, we show that if the signal consists of $k$ significant entries covered by $\n$ iid Gaussian noise, which is the case for many compressed sensing applications, it is possible to derive $\ell_2/\ell_2$ guarantees. }

Calderbank et al. \cite{strip} have considered deterministic sensing matrices that with overwhelming probability act as a near isometry on $k$-sparse vectors, and we refer to this geometric property as the \textit{Statistical} Restricted Isometry Property: \begin{definition} {\bf{(}}$\boldsymbol{(k,\epsilon,\delta)}${\bf{-StRIP matrix)}} %\\
 An $\,\mm \times \n$ (sensing) matrix $\A$ is said to be a
$(k,\epsilon,\delta)$-STRIP, if for $k$-sparse
vectors $\as \in \mathbb{R}^{\n}$, the inequalities
\begin{equation}
\mm(1-\epsilon)\,\|\as\|^2 \,\leq \,\left|\!\left|\,\A\as \, \right|\!\right|^2 \,\leq\,\mm(1+\epsilon)\,\|\as\|^2 \,,
\label{strip}
\end{equation}
hold with probability exceeding $1-\delta\,$ (with respect to a uniform distribution of the vectors
$\as$ among all $k$-sparse vectors in $\mathbb{R}^{\n}$ of the same norm).
\end{definition}
The framework includes sensing matrices for which the columns are \textit{discrete chirps} either in the standard Fourier domain \cite{LHSC} or the Walsh-Hadamard domain \cite{HSC}.  

Chirp Reconstruction is similar to Matching Pursuit in that at each iteration it identifies a significant component of the $k$-sparse signal. The overall computational complexity of Chirp Reconstruction applied to Reed Muller sensing matrices is $O (k\m \log^2 \m)$. The StRIP property of the Reed Muller sensing matrices makes it possible to accurately recover the coefficients of the $k$ significant components leading to robust recovery guarantees in the presence of noise both in the data and in the measurement domains. These guarantees apply with overwhelming probability to the class of approximately $k$-sparse signals. 

%% file: srm.tex
Here $m$ is odd, the rows of the sensing matrix $\A$ are indexed by binary $m$-tuples $x$, and the columns are indexed by pairs $P,b$, where $P$ is an $m\times m$ binary symmetric matrix and $b$ is a binary $m$-tuple. The entry $\varphi_{P,b}(x)$ is given by 
\begin{equation}
\label{kerdock}
\varphi_{P,b}(x)=i^{wt(d_P)+2wt(b)}i^{xPx^\top+2bx^\top}
\end{equation}
where $d_p$ denotes the main diagonal of $P$, and $wt$ denotes the \textit{Hamming weight}( the number of $1$s in the binary vector). 

The Delsarte-Goethals set $DG(m,r)$ is a binary vector space containing $2^{(r+1)m}$ binary symmetric matrices with the property that the difference of any two distinct matrices has rank at least $m-2r$ (See \cite{H}). The Delsarte-Goethals sets are nested:
$$DG(m,0)\subset DG(m,1) \subset \cdots \subset DG(m,\frac{(m-1)}{2}).$$

The first set $DG(m,0)$ is the classical Kerdock set, and the last set $DG(m,\nicefrac{(m-1)}{2})$ is the set of all binary symmetric matrices. The $r${th} Delsarte-Goethals sensing matrix is determined by $DG(m,r)$ and has $\m=2^m$ rows and $\n=2^{(r+2)m}$ columns.  and the column sums in the $r^{th}$ Delsarte-Goethals sensing matrix satisfy
\begin{equation}\left| \sum_x \varphi_{P,b}(x)\right|^2= 0~\mbox{or}~\m^{2-\nicefrac{t}{m}}~\mbox{for some }t\in \{m-2r, \cdots, m\} \label{st3}.\end{equation}
We will use the following lemmas which characterize the properties of the Delsarte-Goethals matrices. For detailed proofs see \cite{strip}.
\begin{lemma}\label{group} Let ${\cal G}={\cal G}(m,r)$ be the set of column vectors $\varphi_{P,b}$ where 
$$\varphi_{P,b}(x)=i^{wt(d_P)+2wt(b)} i^{xPx^\top+2bx^\top}~,~\mbox{for }x\in \mathbb{F}_2^m$$
where $b\in \mathbb{F}_2^m$ and where the binary symmetric matrix $P$ varies over the Delsarte-Goethals set $DG(m,r)$. Then ${\cal G}$ is a group of order $2^{(r+2)m}$ under pointwise multiplication.
\end{lemma}
The following Theorem has been proved by Calderbank et.al.
\begin{theorem}\label{maintheorem}
Suppose the $\mm \times \n$ matrix  $\A$ is derived from a $DG(m,r)$ family, and let $\eta=1-\nicefrac{2r}{m}$.
Then for any $k,\,\epsilon$ with $k \,<\,1\,+\, (\n\,-\,1)\,\epsilon\,$, $\A$ is $(k,\epsilon,\delta)$-StRIP with 
%\begin{enumerate}
%\item
$\delta\,:=\,2\exp\left[\,-\,\frac{{{[\epsilon-(k-1)/(\n-1)]}}^2\,\mm^{\eta}}{32\, k}\,\right]$.
%\item
%If $k^2 \,\geq\,\n$ and $\mm \,\geq\,c\,\frac{k^{1+\mbox{\rm{\footnotesize{o}}}(1)} \, \log\n}{\epsilon^2}$, then $\A$ is $(k,\epsilon,\delta)$-UStRIP with $\delta\,:=\,\,2\exp\left[\,-\,\frac{\epsilon^2\,\mm^{\eta}}{2\, k}\,\right]$.
%\end{enumerate}
\label{mainthm}
\end{theorem}

%\begin{theorem}\label{tab}
%Let $Q$ be a binary symmetric $m\times m$ matrix, and let $b \in \mathbb{F}_2^m$. If $$S=\sum_x i^{xQx^\top+2bx^\top}$$ then either $S=0$ or 
%$$S^2=i^{z_1Qz_1^\top+2bz_1^T} 2^{2m-t}~~,~~~\mbox{where }z_1Q=d_Q.$$
%\end{theorem} 
%\begin{corollary}
%\label{cor}
%Let $\Phii{P_1,b_1}$ and $\Phii{P_2,b_2}$ be two distinct columns of the $r$th DG sensing matrix. If $S=\sum_x \overline{\Phii{P_1,b_1}}\Phii{P_2,b_2}$  then $|S|\leq N^{\frac{1}{2}+\frac{r}{2m}}$. 
%\end{corollary}

%% file: quad-alg.tex
In this section we introduce the Chirp Reconstruction Algorithm, used for the purpose of efficient sparse reconstruction in the presence of noise. Let $\pi=\{\pi_1,\cdots,\pi_\n\}$ be a random permutation of $\{1,\cdots,\n\}$, and let $\as$ be an almost $k$-sparse vector whose $k$ significant entries are positioned according to $\{\pi_1,\cdots,\pi_k\}$. Let $\as_k$ be $\as$ restricted to its best $k$-term approximation. Calderbank et.al. showed that if $\A$ is $(k,\epsilon,\delta)$ StRIP, then with probability $1-\delta$, \begin{equation}\|\A(\as-\as_k)\|_2 \leq \|\as-\as_k\|_1.\label{l2l1}\end{equation} Furthermore, if we assumed that $\as$ is exactly $k$-sparse encompassed with $\n$ iid white noise with variance $\sigma_\n^2$, then since the rows of $\A$ form a tight-frame with redundancy  $\nicefrac{\n}{\mm}$, it follows that noise samples on distinct measurements are independent gaussian, with variance $\nicefrac{\n\sigma_\n^2}{\mm}$. As a result, using the concentration bounds for $\chi^2$ distribution, it follows that with overwhelming probability \begin{equation}\label{l2l2}\|\nor\A(\as-\as_k)\|_2 \leq \|\as-\as_k\|_2\end{equation}

Let $\mu$ be the noise in the measurement domain. Then compressive sensing using the matrix $\nor\A$ maps a vector $\as$ to $$f=\nor\A\as+\mu=y+\nu,$$
where $y=\nor\A\as_k$, and $\nu=\nor\A(\as-\as_k)+\mu$. The goal is then to approximate $\as_k$ from $f$. The chirp reconstruction algorithm \cite{LHSC,HSC}  is a repurposing of the chirp detection algorithm commonly used in navigation radars which is known to work extremely well in the presence of noise, and is described as Algorithm \ref{alg1}. At each iteration $t$, given the residual measurement vector $\f_t$, first the autocorrelation function is applied to $f_t$, \textit{i.e} $f_t$ is pointwise multiplied with a shifted version of itself. Then applying the fast Hadamard transform forms the power spectrum of $f_t$, which as we will show, consists of $k$ tones corresponding to the position of the $k$ significant entries of $\as$, and a noise term uniformly spread across all Hadamard coefficients, which accounts for the noise $\nu$, and chirp like cross-terms. In other words, since the sensing matrix is obtained by exponentiating quadratic functions, forming the power spectrum produces a sparse superposition of pure frequencies (in the example below, these are Walsh functions in the binary domain) against a background of chirp-like cross terms. The algorithm then iteratively learns the terms in the sparse superposition by varying the offset $a$. These terms can be peeled off in decreasing order of signal strength or processed in a list. Experimental results show close approach to the information theoretic lower bound on the required number of measurements \cite{HSC}.

\begin{algorithm}
\caption{Chirp Reconstruction Algorithm}
   Input: $\m$ dimensional vector $f^1=\nor\A\as_k+\nu$,
   Output: An approximation $\has^*$ to the $k$-sparse signal $\as_k$
     \label{alg1}
   \begin{algorithmic}[1]
  \FOR{$t=1,\cdots,k$ or while $\|\boldsymbol{f^t}\|_2 \geq \epsilon$}
   \FOR{ $j=1,\cdots,m$}
   \STATE Let $a_j$ be the $j$th standard basis vector. Using $a_j$ pointwise multiply $f_t$ with its shifted vector.
   \STATE Compute the fast Walsh-Hadamard transform of the computed auto-correlation: Equation~(\ref{fourier}).
   \STATE Find the position of the next peak $l_{t,j}$ in the Hadamard domain.Decode the next row of the $j^{th}$ row of $P_{\pi_t}$.
      \ENDFOR 
      \STATE Pointwise multiply $f^t$ with $i^{xP_{\pi_t}x^\top}$, and find the corresponding value $b_{\pi_t}$, by finding the next peak in the power spectrum.
  \STATE  \label{opt} Determine the corresponding value $\has^+_{\pi_t}$ which minimizes $\|\sqrt{N}f^t- \has_{\pi_t} \varphi_{P_{\pi_t},b_{\pi_t}}\|^2$.  
   \STATE Set ${f^{t+1}}\doteq {f^t}- \has^+_{\pi_t} \varphi_{P_{\pi_t},b_{\pi_t}}$. 
   \ENDFOR
   \STATE Let $\A_{\pi_1^k}$ be $\A$ restricted to the recovered $k$ columns. Output $\has^*\doteq \arg\min\|\nor\A_{\pi_1^k}\has-f\|^2.$
   \end{algorithmic}
   \end{algorithm}
\vskip-0.07cm
The first step is pointwise multiplication of the sparse superposition with a shifted copy of itself, which gives \begin{equation}\at{y}+\at{\nu}+y(x+a)\overline{\nu(x)}+\nu(x+a)\overline{y(x)}\label{auto} \end{equation} 
By Cauchy-Schwartz inequality and StRIP propery, it is easy to verify that the total energy of the last three terms in (\ref{auto}) is bounded by $3\|\nu\|^2\|\as_k\|^2$. The first term itself can be decomposed into pure tones $\frac{1}{\m}\sum_{j=1}^k |\alpha_j|^2(-1)^{a^\top P_{\pi_j}x}$, and chirp terms  \begin{equation}\label{chirp1}\frac{1}{\m}\sum_{i\neq j} \alpha_i\overline{\alpha_j} \varphi_{P_{\pi_i},b_{\pi_i}}(x+a)\overline{\varphi_{P_{\pi_j},b_{\pi_j}}(x)}.\end{equation}

Then the (fast) Hadamard transform concentrates the energy associated with pure tones into (at most) $k$ Walsh-Hadamard tones with energies $|\alpha_j|^4$. This algorithm may get into trouble when two of the pure tones fall into the same basis. This problem can be resolved to a large extent by varying the offset $a$ \cite{HSC}. In the next section, we show that the the the fast Hadamard transform distributes the energy of Equation~(\ref{chirp1}) uniformly across all $\m$ tones in the fast Hadamard domain. Moreover, by Azuma's inequality, it is easy to verify that the total energy of the chirps terms (Equation~(\ref{chirp1})) is with high-probability at most $\frac{2\sum_{i\neq j}|\alpha_i||\alpha_j|}{\mm^2}$. The impact of reducing the signal strength in the $k$ concentrated peaks which does not make a problem in detecting the largest peak in the presence of sufficiently large SNR.

%% file: matrices.tex
The $l^{th}$ Fourier coefficient of the term~(\ref{chirp1}) is
\begin{equation}
\label{fourier}
\Gamma_a^l=\frac{1}{\m^{\nicefrac{3}{2}}} \sum_{j\neq t} \alpha_j \overline{\alpha_t} \sum_x (-1)^{l^\top x} \varphi_{P_{\pi_j},b_{\pi_j}}(x+a) \overline{\varphi_{P_{\pi_t},b_{\pi_t}}(x)}.
\end{equation}
In this section we show that with overwhelming probability, for all Fourier coefficients $l$, $\left|\Gamma_a^l\right|\leq {\sqrt{\frac{k}{N^\eta}}}\|\as_k\|^2$, where the probability is with respect to the permutation $\pi$. We show this by a probabilistic argument. First we show that $\mathbb{E}_\pi\left[ \left|\Gamma_a^l\right| \right]= 0$, and then by constructing an appropriate martingale sequence, and applying the Azuma's inequality we show that $\left|\Gamma_a^l\right|$ is highly concentrated around its expectation. 

 Let $\cal T$ be the set of all $k$-tuples $(t_1,\cdots,t_k)$, such that $\{t_1,\cdots,t_\n\}$ is a permutation of $\{1,\cdots,\n\}$. For all distinct $i,j$ in $\{1,\cdots,k\},$ and $(t_1,\cdots,t_k)$ in $\cal T$ define
\begin{equation}
\label{h}
h(t_i,t_j)\doteq\sum_x (-1)^{\ell x^\top} \varphi_{\pii{t_i}}(x+a) \overline {\varphi_{\pii{t_j}}(x)},
\end{equation}and
\begin{equation}
\label{g} \g(t_1,\cdots,t_k)\doteq\frac{1}{\m^{\frac{3}{2}}} \sum_{i\neq j} \alpha_i \overline{\alpha_j} h(t_i,t_j),
\end{equation}
Then (\ref{fourier}) can be written as $\g(\pi_1,\cdots,\pi_k)$. We first show that $\mathbb{E}_\pi\left[\left|\g(\pi_1,\cdots,\pi_k)\right|\right]= 0$.
\begin{lemma}
\label{l1}
Let $\cal G$ be the group of columns of $\A$ with respect to pointwise multiplication. The map ${\cal G}\times {\cal G}\rightarrow \{\pm1,\pm i\}$ given by $(\bs{g},\bs{h})\rightarrow \bs{g}(x+a)\bs{h}^{-1}(x)$ is a surjective homomorphism, and
$$\sum_{\bs{g}\neq\bs{h}}\bs{g}(x+a)\bs{h}^{-1}(x)=-\sum_{\bs{g}}\bs{g}(x+a)\bs{g}^{-1}(x).$$ 
\end{lemma}
\begin{proof}
$\sum_{\bs{g},\bs{h}}\bs{g}(x+a)\bs{h}^{-1}(x)=0$.
\end{proof}
\begin{lemma}\label{exp}
$\Ex_\pi\left[ \g(\pi)\right]$ is zero. %either zero or decays to zero at the rate $\frac{\m}{\n}$.
\end{lemma}
\begin{proof}
We can rewrite $$\Ex_{{\substack{\pi\\i \neq j}}}[\sum_x (-1)^{\ell x^\top} \varphi_{\pii{\pi_i}}(x+a) \overline {\varphi_{\pii{\pi_j}}(x)}]$$ in the form
\begin{equation}\label{form}\frac{1}{\n(\n-1)}\sum_x (-1)^{\ell x^\top} \sum_{\bs{g}\neq \bs{h}} \bs{g}(x+a)\bs{h}^{-1}(x).\end{equation}
The initial factor is just the frequency with which any admissible pair is chosen, and the second sum is taken over the column group $\cal G$. Lemma~\ref{l1} allows us to rewrite (\ref{form}) as 
\begin{eqnarray}\nonumber
\label{form2}& &\frac{-1}{\n(\n-1)} \sum_x (-1)^{\ell x^\top} \sum_{\bs{g}} \bs{g}(x+a)\bs{g}^{-1}(x)\\ &=&\frac{-1}{\n(\n-1)} \sum_P i^{aPa^\top}  \sum_x (-1)^{(aP+\ell) x^\top}  \sum_b (-1)^{ab^\top},
\end{eqnarray}
where the outer sum is taken over all binary symmetric matrices in the \srm ensembles. Since $a\neq 0$, the sum $\sum_b (-1)^{ab^\top}=0$ is always zero%, and if $a=0$ then (\ref{form2}) reduces to 
%$$ \frac{-1}{\n-1} \sum_x (-1)^{\ell x^\top}=
%\left\{
%\begin{array}{cc}
 %\frac{-\m}{\n-1} & \mbox{if $\ell=0$}    \\
 % 0& \mbox{otherwise}      
%\end{array}
%\right.
 % $$
\end{proof}
\begin{theorem}\label{md2}
Let $\pi$ be a random permutation of $\{1,\cdots,\n\}$. Then with probability at least $1-\delta$ for any coefficient $l$ we have 
\begin{equation}
 \g(\pi_1,\cdots,\pi_k)\leq \sqrt{\frac{8k\log\left(\frac{\m}{\delta}\right)}{\m^{1-r/m}}}\|\bs{\alpha}\|^2.
\end{equation}
\end{theorem}
\begin{IEEEproof}
Define the martingale sequence $Z_1,\cdots,Z_k$ as \begin{equation}\label{martingale}Z_i=\mathbb{E}_\pi\left[ \g(\pi_1,\cdots,\pi_k)\left|\right.\pi_1,\cdots,\pi_i\right],\end{equation}   
and denote $\pi_i^j\doteq(\pi_i,\cdots,\pi_j)$. Since the columns of $\A$ form a group under pointwise multiplication, using Equation~(\ref{st3}) we get
\begin{eqnarray}
\nonumber
&&\left|\sup_u \mathbb{E}_\pi\left[ \g(\pi_1^k)\left|\right.\pi_1^{i-1},u\right]-\inf_l \mathbb{E}_\pi\left[ \g(\pi_1^k)\left|\right.\pi_1^{i-1},l\right]\right| \\ \label{marting2}&\leq& \frac{|\alpha_i| |\sum_{j\neq i}\alpha_j|}{\m^{\frac{m-r}{m}}}.
\end{eqnarray}
Note that by Cauchy-Schwartz inequality
$$\sum_i \left(|\alpha_i| |\sum_{j\neq i}\alpha_j|^2 \right)\leq k\left(\sum_i |\alpha_i|^2\right)^2.$$
Consequently, by applying Azuma's inequality we get 
$$\Pr_\pi\left[  \g(\pi_1,\cdots,\pi_k)\geq \epsilon  \right] \leq \exp\left(\frac{-\m^{1-r/m}\epsilon^2}{8k \|\as\|_2^4}\right).
$$
Applying the union bounds on all $\m$ possible choices of $l$ completes the proof. 
\end{IEEEproof}

%% file: l2.tex
Consequently, the chirp-like terms have uniform distribution across all $N$ tones in the fast hadamard domain. Consequently, if $k\ll \n$, and the SNR is sufficiently large, it is possible to iteratively recover the positions of the $k$ significant entries of the vector $\as$. Having recovered the support $\pi_1^k$ of $\as_k$, it is possible to reconstruct a better approximation for $\as_k$ by minimizing $\|\nor\A_{pi_1^k}\hat-\f\|^2$, which has the analytical solution \begin{equation}\has^*\doteq\sqrt{\mm}\left(\A_{\pi_1^k}^\dag\A_{\pi_1^k}\right)^{-1}\A_{\pi_1^k}^\dag f.\label{el2}\end{equation} The following bound on the approximation error of $\has^*$ then follows from the StRIP property.
\begin{theorem}
Let $\A$ be $(k,\epsilon,\delta)$-StRIP. Let $\as$ be an almost $k$-sparse vector such that $\as_k$ has a uniformly random support $\{\pi_1,\cdots,\pi_k\}$. Let $\has^*$ defined by Equation~(\ref{el2}). Then with probability $1-\delta$, 
$$\|\has^*-\as_k\|_2\leq \frac{2}{(1-\epsilon)}\left(\nor\|\A(\as-\as_k)\|_2+\|\mu\|_2\right).$$
\end{theorem}
\begin{proof}
Since $\A$ is $(k,\epsilon,\delta)$-StRIP, and $\as_k$ and $\has^*$ are two $k$-sparse vectors with the same random support, with probability $1-\delta,~(1-\epsilon)\|\has^*-\as_k\|_2\leq\nor\|\A(\has^*-\as_k)\|_2$. By the triangle inequality
$$\nor\|\A(\has^*-\as_k)\|_2\leq \|\nor\A\has^*-\f\|_2+\|\nu\|_2.$$ 
On the other hand, by definition of $\has^*$ we have
$$ \|\nor\A\has^*-f\|_2\leq \|\nor\A\as_k-f\|_2\leq \|\nu\|_2.$$
Putting all together, and recalling that $$\|\nu\|_2\leq \nor\|\A(\as-\as_k)\|_2+\|\mu\|_2$$
Completes the proof.\end{proof}
As a result, it follows from Equation~(\ref{l2l1}), that with probability at least $1-2\delta$,
 $$\|\has^*-\as_k\|_2\leq \frac{2}{(1-\epsilon)}\left(\nor\|\as-\as_k\|_1+\|\mu\|_2\right),$$ and furthermore, considering Equation~(\ref{l2l2}), if the signal in the data domain consists of $k$-significant entries covered by white noise, then with overwhelming probability
 $$\|\has^*-\as_k\|_2\leq \frac{2}{(1-\epsilon)}\left(\|\as-\as_k\|_2+\|\mu\|_2\right).$$